\documentclass[a4paper,11pt]{article}

\usepackage{microtype}

\usepackage{fullpage}

\usepackage[utf8]{inputenc}
\usepackage[T1]{fontenc}
\usepackage[english]{babel}

\usepackage{amsmath,amssymb,amsfonts,amsthm}
\usepackage{mathabx}

\usepackage{cleveref}

\newtheorem{lemma}{Lemma}
\newtheorem{theorem}{Theorem}

\theoremstyle{definition}
\newtheorem{definition}{Definition}

\usepackage[pdftex]{graphicx}
\usepackage[usenames,dvipsnames,table]{xcolor}

\usepackage{booktabs}

\def\polylog{\operatorname{polylog}}
\newcommand{\ceil}[1]{\left\lceil{#1}\right\rceil}
\newcommand{\floor}[1]{\left\lfloor{#1}\right\rfloor}
\newcommand{\eps}{\varepsilon}
\newcommand{\calG}{\mathcal G}
\newcommand{\bbZ}{\mathbb Z}
\newcommand{\set}[1]{\left \{ #1 \right \}}
\newcommand{\uc}{\mbox{\rm uc}}
\newcommand{\Prp}[1]{\Pr\!\left[{#1} \right]}
\newcommand{\E}{\mathbf{E}}
\newcommand{\Ep}[1]{\E\!\left[{#1} \right]}

\newcommand*\samethanks[1][\value{footnote}]{\footnotemark[#1]}

\usepackage[usenames,dvipsnames]{xcolor}

\usepackage{todonotes}

\usepackage{authblk}

\title{Sublinear Distance Labeling}

\author[1]{Stephen Alstrup\thanks{Research partly supported by the FNU
project AlgoDisc - Discrete Mathematics, Algorithms, and Data
Structures.}}
\author[1]{Søren Dahlgaard\thanks{Research partly supported by Mikkel Thorup's
        Advanced Grant from the Danish Council for Independent Research
        under the Sapere Aude research career programme.}}
\author[1]{Mathias Bæk Tejs Knudsen\samethanks[1]\samethanks[2]}
\author[2]{Ely Porat}
\affil[1]{University of Copenhagen\\\texttt{\{s.alstrup,soerend,knudsen\}@di.ku.dk}}
\affil[2]{Bar-Ilan University\\\texttt{porately@cs.biu.ac.il}}

\begin{document}

\maketitle
\begin{abstract}
A distance labeling scheme labels the $n$ nodes of a graph with binary strings
such that, given the labels of any two nodes, one can determine the distance in
the graph between the two nodes by looking only at the labels. A
$D$-preserving distance labeling scheme only returns precise distances between
pairs of nodes that are at distance at least $D$ from each other. In this paper
we consider distance labeling schemes for the classical case of unweighted
graphs with both \emph{directed} and \emph{undirected} edges.

We present a $O(\frac{n}{D}\log^2 D)$ bit $D$-preserving
distance labeling scheme, improving the previous bound by Bollobás et. al.
[SIAM J. Discrete Math. 2005]. We also give an almost matching lower bound of
$\Omega(\frac{n}{D})$.
With our $D$-preserving distance labeling scheme as a building block, we
additionally achieve the following results:

\textbf{1.} We present the first distance labeling scheme of size $o(n)$ for
sparse graphs
(and hence bounded degree graphs). This addresses an open problem by Gavoille
et. al. [J. Algo. 2004], hereby separating the complexity from distance
labeling in general graphs which require $\Omega(n)$ bits, Moon [Proc. of Glasgow Math. Association 1965].
\footnote{This result for sparse graphs was made available online in
    a preliminary version of this paper \cite{arxivVersionThisPaper}.
    The label size was subsequently slightly improved by an $O(\log \log n)$
factor by Gawrychowski et al. \cite{GawrychowskiKUEvenSimpler}.}

\textbf{2.} For approximate $r$-additive labeling schemes, that return distances
within an additive error of $r$
we show a scheme of size
    $O\left ( \frac{n}{r} \cdot\frac{\polylog (r\log n)}{\log n} \right )$ for
$r \ge 2$. This improves on the current best bound of $O\left(\frac{n}{r}\right)$
    by Alstrup et. al. [SODA 2016] for sub-polynomial $r$, and is a generalization of a result by Gawrychowski
et al. [arXiv preprint 2015] who showed this for $r=2$.
\end{abstract}

\newpage
\section{Introduction}

The concept of \emph{informative labeling schemes} dates back to Breuer and
Folkman \cite{Breuer66,BF67} and was formally introduced by Kannan et
al.~\cite{KNR92,muller}. A labeling scheme is a way to represent a graph in a
distributed setting by assigning a bit string (called a \emph{label}) to each
node of the graph. In a distance labeling scheme we assign labels to a
graph $G$ from a family $\mathcal{G}$ (i.e. all forests, bounded-degree graphs,
or planar graphs with $n$ nodes) such that, given \emph{only} the labels of
a pair of nodes, we can compute the distance between them without the need for
a centralized data structure. When designing a labeling scheme the main goal is
to minimize the \emph{maximum label size} over all nodes of all graphs $G$ in
the family $\mathcal{G}$. We call this the size of the labeling scheme. As a
secondary goal some papers consider the \emph{encoding} and \emph{decoding} time
of the labeling scheme in various computational models. In this paper we study
the classical case of \emph{unweighted} graphs.

\paragraph{Exact distances}
The problem of exact distance labeling in general graphs is a classic
problem that was studied thoroughly in the 1970/80's. Graham and Pollak
\cite{grahampollak} and Winkler \cite{winkler} showed that labels of size
$\ceil{(n-1) \cdot \log_2 3}$ suffice in this case. Combining~\cite{KNR92}
and~\cite{moon1965minimal} gives a lower bound of $\ceil{n/2}$ bits
(see also~\cite{Gavoille200485}). Recently, Alstrup et al.~\cite{distalstrup}
improved the label size to $\frac{\log_2 3}{2}n + O\!\left(\log^2 n\right)$ bits.

Distance labeling schemes have also been investigated for various families
of graphs, providing both upper and lower bounds. For trees, Peleg~\cite{peleg}
showed that labels of size $O(\log^2 n)$ suffice with a matching lower bound by
Gavoille et. al~\cite{Gavoille200485}.
Gavoille et. al~\cite{Gavoille200485} also showed a $\Omega(n^{1/3})$ lower
bound for planar graphs and $\Omega(\sqrt{n})$ bound for bounded degree (and
thus sparse) graphs and for weighted graphs this was recently improved by
Abboud and Dahlgaard~\cite{AbboudD16} to $\Omega(\sqrt{n}\log n)$ even for
bounded planar graphs. Gavoille et al.~\cite{Gavoille200485} provided an
$O(\sqrt{n}\log n)$ labeling scheme for planar graphs (even weighted), however
nothing better than the $O(n)$ scheme for general graphs is known for
bounded-degree graphs. It has remained a major open problem in the field of
labeling schemes whether a scheme of size $O(\sqrt{n})$ or even $o(n)$ exists
for bounded-degree graphs as stated in e.g.~\cite{Gavoille200485}.

Other families of graphs studied include distance-hereditary~\cite{GP03b},
bounded clique-width~\cite{CV03}, some non-positively curved
plane~\cite{CDV06}, as well as interval~\cite{GP08} and permutation
graphs~\cite{BG09}.

\paragraph{Approximate distances}
For some applications, the $\Omega(poly(n))$ requirement on the label size for several graph classes is prohibitive. Therefore a large
body of work is dedicated to labeling schemes for approximating distances in
various families of
graphs~\cite{AbrahamCG12,CDEHV08,ElkinFN15,Gavoille200485,GuptaKL03,gupta2005traveling,KL06,peleg,Talwar04,Thorup2004distance,ThZw05}.
Such labeling schemes often provide efficient implementations of other
data structures like distance oracles \cite{Thorup2004distance} and dynamic
graph algorithms \cite{AbrahamCG12}.

In~\cite{peleg} a labeling scheme of size $O(\log^2 n\cdot \kappa\cdot
n^{1/\kappa})$ was presented for approximating distances up to a factor\footnote{This does not break the Girth Conjecture, as the
labeling scheme may under-estimate the distance as well.} of
$\sqrt{8\kappa}$.
In~\cite{Thorup2004distance} a scheme of poly-logarithmic
size was given for planar graphs when distances need only be reported within a
factor of $(1+\eps)$. Labeling schemes of additive error have also been
investigated. For general graphs Alstrup et. al~\cite{distalstrup} gave a
scheme of size $O(n/r)$ for $r$-additive distance labeling with $r\ge 2$ and a
lower bound of $\Omega(\sqrt{n/r})$ was given by Gavoille et
al.~\cite{GKKPP01}.
For $r=1$ a lower bound of $\Omega(n)$ can be
established by observing that such a scheme can answer adjacency queries in
bipartite graphs, which require $n/4$ bits to label for adjacency.

\paragraph{Distance preserving}
An alternative to approximating all distances is to only report exact
distances above some certain threshold $D$. A labeling scheme, which reports
exact distances for nodes $u,v$ where $dist(u,v)\ge D$ is called a
\emph{$D$-preserving distance labeling scheme}\footnote{In this paper we adopt
the convention that the labeling scheme returns an upper-bound if
$dist(u,v) < D$.}. Bollobás et al.~\cite{BCE03} introduced this notion and
gave a labeling scheme of size $O(\frac{n}{D}\log^2 n)$ for both directed and
undirected graphs. They also provided an $\Omega(\frac{n}{D}\log D)$ lower
bound for directed graphs.

\subsection{Related work}\label{secondmatter}
A problem closely related to distance labeling is adjacency labeling. For
some classes such as general graphs the best-known lower bounds for
distance is actually that of adjacency. Adjacency labeling has been studied
for various classes of graphs. In~\cite{AlstrupKTZ14} the label size for
adjacency in general undirected graphs was improved from $n/2+O(\log
n)$~\cite{KNR92,moon1965minimal} to optimal size $n/2+ O(1)$, and
in~\cite{adjacencytrees2015} adjacency labeling for trees was improved from
$\log_2 n + O(\log^* n)$~\cite{alstruprauhe} to optimal size $\log_2 n + O(1)$.

Distance labeling schemes and related $2$-hop labeling are used in SIGMOD and
is central for some real-world applications~\cite{AIY13,DGSW14,JRXL12}.
Approximate distance labeling schemes have found applications in several fields
such as reachability and distance oracles \cite{Thorup2004distance} and
communication networks \cite{peleg}. An overview of distance labeling schemes
can be found in~\cite{distalstrup}.

\subsection{Our results}
We address open problems of \cite{distalstrup,BCE03,Gavoille200485}
improving the label sizes for \emph{exact distances in sparse graphs},
\emph{$r$-additive distance in general graphs}, and \emph{$D$-preserving
distance labeling}. We do this by showing a strong relationship between
$D$-preserving distance labeling and several other labeling problems using
$D$-preserving distance labels as a black box. Thus, by improving the result of
\cite{BCE03} we are able to obtain the first sublinear labeling schemes for
several problems studied at SODA over the past decades. Our results hold for
both directed and undirected graphs and are summarized below.

\paragraph{Sparse graphs:}
We present the first sublinear distance labeling scheme for sparse graphs
giving the following theorem:
\begin{theorem} \label{thm:sparse}
    Let $\mathcal{S}_n$ denote the family of unweighted
    graphs on $n$ nodes with at most $n^{1+o(1)}$ edges. Then there
    exists a distance labeling scheme for $\mathcal{S}_n$ with maximum label
    size $o(n)$.
\end{theorem}
As noted, prior to this work the best-known bound for this family was the
$O(n)$ scheme of \cite{distalstrup} for general graphs.
Thus, Theorem~\ref{thm:sparse} separates the
family of sparse graphs from the family of general graphs requiring $\Omega(n)$
label size.
Our result uses a black-box reduction from sparse graphs to the $D$-preserving
distance scheme of Theorem~\ref{thm:d_dist} below. The result of
Theorem~\ref{thm:sparse} was made available online in a preliminary version of
this paper~\cite{arxivVersionThisPaper} and was subsequently slightly improved
by Gawrychowski et al. \cite{GawrychowskiKUEvenSimpler} by noting, that one of
the steps in the construction of our $D$-preserving distance scheme can be
skipped when only considering sparse graphs\footnote{
The scheme presented in this paper has labels of length
$O \left ( \frac{n \polylog \Delta}{\Delta} \right )$, where
$\Delta = \frac{\log n}{1+\log \frac{m+n}{n}}$. In
\cite{GawrychowskiKUEvenSimpler} they improve the exponent
        of the $\polylog \Delta$ term from $2$ to $1$.}.

\paragraph{Approximate labeling schemes:}
For $r$-additive distance labeling
Gawrychowski et al.~\cite{GawrychowskiKUEvenSimpler} showed that a sublinear
labeling scheme for sparse graphs implies a sublinear labeling scheme for $r=2$
in general graphs. We generalize this result to $r \ge 2$ by a reduction
to the $D$-preserving scheme. We note that a reduction to sparse graphs
does not suffice in this case, and the scheme of
\cite{GawrychowskiKUEvenSimpler} thus only works for $r=2$. More precisely, we
show the following:

\begin{theorem}\label{thm:r_approx}
    For any $r \ge 2$, there exists an approximate $r$-additive labeling
    schemes for the family $\calG_n$ of unweighted graphs on $n$ nodes with
    maximum label size
    \[
        O\!\left(\frac{n}{r}\cdot \frac{\polylog (r\log n)}{\log n} \right)\ .
    \]
\end{theorem}

Theorem~\ref{thm:r_approx} improves on the previous best bound of $O\left(\frac{n}{r}\right)$
by \cite{distalstrup} whenever $r = 2^{o\left (\sqrt{\log n} \right )}$, e.g.
when $r = \polylog n$.

\paragraph{$D$-preserving labeling schemes:}
For $D$-preserving labeling schemes we show that:
\begin{theorem}\label{thm:d_dist}
    For any integer $D \in [1,n]$, there exists a $D$-preserving distance
    labeling scheme for the family $\calG_n$ of unweighted
    graphs on $n$ nodes with maximum label size
    \[
        O\!\left(\frac{n}{D}\max\set{\log^2 D,1}\right)\ .
    \]
\end{theorem}

Theorem~\ref{thm:d_dist} improves the result of \cite{BCE03} by a factor of
$O(\log^2 n/\log^2 D)$ giving the first sublinear size labels for this
problem for any $D = \omega(1)$. This sublinearity is the main ingredient in
showing the results of Theorems~\ref{thm:sparse} and~\ref{thm:r_approx}. Our
scheme uses sampling similar to that of \cite{BCE03}. By sampling
fewer nodes we show that not ``too many'' nodes end up being problematic and
handle these separately\footnote{We note that after making this result
available online in a preliminary version~\cite{arxivVersionThisPaper}, the
bound of Theorem~\ref{thm:d_dist} was slightly improved by Gawrychowski et
al.~\cite{GawrychowskiKUEvenSimpler} to $O(\frac{n}{D}\log D)$.}.

Finally, we show the following almost matching lower bound for undirected
graphs extending the construction of~\cite{BCE03} for directed graphs.

\begin{theorem} \label{thm:lowerapprox}
A $D$-preserving distance labeling scheme for the family $\calG_n$ of
unweighted and undirected
graphs on $n$ nodes require label size $\Omega(\frac{n}{D})$, when $D$ is an
integer in $[1,n-1]$.
\end{theorem}

\section{Preliminaries}\label{sec:prelims}
Throughout the paper we adopt the convention that $\lg x = \max(\log_2 x, 1)$
and $\log x = \ln x$. When $x \le 0$ we define $\lg x = 1$.
In this paper we assume the word-RAM model, with word size $w = \Theta(\log
n)$. If $s$ is a bitstring we denote its length by $|s|$ and will also use $s$
to denote the integer value of $s$ when this is clear from context. We use
$s\circ s'$ to denote concatenation of bit strings. Finally, we use the
Elias $\gamma$ code~\cite{Elias75} to encode a bitstring $s$ of unknown length
using $2|s|$ bits such that we may concatenate several such bitstrings and
decode them again.

\paragraph{Labeling schemes}\label{sec:prelim_labeling}
A \emph{distance labeling scheme} for a family of graphs $\calG$ consists of an
encoder $e$ and a decoder $d$. Given a graph $G\in\calG$ the encoder computes a
\emph{label assignment} $e_G : V(G)\to \{0,1\}^*$, which assigns a \emph{label}
to each node of $G$. The decoder is a function such that given any graph
$G\in\calG$ and any pair of nodes $u,v\in V(G)$ we have $d(e_G(u), e_G(v)) =
dist_G(u,v)$. Note that the decoder is oblivious to the actual graph $G$ and is
only given the two labels $e_G(u)$ and $e_G(v)$.

The \emph{size} of a labeling scheme is defined as the maximum label size
$|e_G(u)|$ over all graphs $G\in\calG$ and all nodes $u\in V(G)$. If for all
graphs $G\in\calG$ the mapping $e_G$ is injective we say that the labeling
scheme assigns \emph{unique labels} (note that two different graphs $G,
G'\in\calG$ may share a label).

If the encoder and graph is clear from the context, we will sometimes denote the
label of a node $u$ by $\ell(u) = e_G(u)$.

Various computability requirements are sometimes imposed on labeling
schemes~\cite{AKM01,KNR92,siamcompKatzKKP04}.

\section{$D$-preserving distance labeling schemes}\label{sec:preserve}

In this section we will prove Theorem~\ref{thm:d_dist}. Observe first that for $D=1$
Theorem~\ref{thm:d_dist} is exactly the classic problem of distance labeling and we
may use the result of \cite{distalstrup}. We will therefore assume that
$D\ge 2$ for the remainder of this paper. Let us first formalize the definition
of a $D$-preserving distance labeling scheme.

\begin{definition}
    Let $D$ be a positive integer let $\calG$ be a family of graphs. For each
    graph $G\in\calG$ let $e_G : V(G)\to \{0,1\}^*$ be a mapping of nodes to
    labels. Let $d : \{0,1\}^*\times \{0,1\}^*\to \bbZ$ be a decoder. If $e$
    and $d$ satisfy the following two properties, we say that the pair $(e,d)$
    is a \emph{$D$-distance preserving labeling scheme} for the graph family
    $\calG$.
    \begin{enumerate}
        \item $d(e_G(u), e_G(v))\ge dist_G(u,v)$ for all $u,v\in G$ for any
            $G\in\calG$.
        \item $d(e_G(u), e_G(v)) = dist_G(u,v)$ for all $u,v\in G$ with
            $dist_G(u,v)\ge D$ for any $G\in\calG$.
    \end{enumerate}
\end{definition}

The idea of the labeling scheme presented in this section is to first make a
labeling scheme for distances in the range $[D,2D]$ and use this scheme for
increasingly bigger distances until
all distances of at least $D$ are covered. Loosely speaking, the scheme is
obtained by sampling a set of nodes $R$, such that \emph{most} shortest paths
of length at least $D$ contain a node from $R$. Then all nodes are partitioned into
\emph{sick} and \emph{healthy} nodes adding the sick nodes to the set $R$. All
nodes then store their distance to each node of $R$ and healthy nodes will
store the distance to all nodes, for which the shortest path is not
\emph{covered} by some node in $R$.

\subsection{A sample-based approach}\label{sec:old_d_dist_scheme}
As a warm-up, we first present the $O\!\left(\frac{n}{D}\log^2 n\right)$ scheme
of Bollob\'{a}s et al. in \cite{BCE03} with a slight modification.

Given a graph $G=(V,E)\in\calG$ we pick a random multiset $R\subseteq V$
consisting of $\ceil{c\cdot \frac{n}{D}\log n}$ nodes for a constant $c$ to be
decided. Each element of $R$ is picked uniformly and independently at random
from $V$ (i.e. the same node might be picked several times)\footnote{In
\cite{BCE03} they instead picked $R$ by including each node of $G$ with
probability $\frac{c\log n}{D}$.}. We order $R$ arbitrarily as $(w_1,\ldots,
w_{|R|})$ and assign the label of a node $u\in V$ as
\[
    \ell(u) = dist_G(u,w_1)\circ dist_G(u,w_2)\circ \ldots\circ
    dist_G(u,w_{|R|})
\]

\begin{lemma}\label{lem:old_dist_correct}
    Let $u$ and $v$ be two nodes of some graph $G\in\calG$. Set
    \begin{equation}\label{eq:old_dist}
        d = \min_{w\in R} dist_G(u,w) + dist_G(v,w)\ .
    \end{equation}
    Then $d\ge dist_G(u,v)$ and $d = dist_G(u,v)$ if $R$ contains a node from a
    shortest path between $u$ and $v$.
\end{lemma}
\begin{proof}
    Let $z\in R$ be the node corresponding to the minimum value of
    \eqref{eq:old_dist}. We then have $d = dist_G(u,z) + dist_G(z,v)$. By the
    triangle inequality this implies $d\ge dist(u,v)$.

    Now let $p$ be some shortest path between $u$ and $v$ in $G$ and assume
    that $z\in p$. Then $dist_G(u,v) = dist_G(u,z) + dist_G(z,v)$, implying
    that $d \le dist_G(u,v)$, and thus $d = dist(u,v)$.
\end{proof}

By Lemma~\ref{lem:old_dist_correct} it only remains to show that the set $R$ is likely
to contain a node on a shortest path between any pair of nodes $u,v\in V$ with
$dist_G(u,v)\ge D$.

\begin{lemma}\label{lem:old_dist_probability}
    Let $R$ be defined as above. Then the probability that there exists a pair
    of nodes $u,v\in V$ such that $dist_G(u,v)\ge D$ and no node on the
    shortest path between $u$ and $v$ is sampled is at most $n^{2-c}$.
\end{lemma}
\begin{proof}
    Consider a pair of nodes $u,v\in V$ with $dist_G(u,v)\ge D$. Let $p$ be a
    shortest path between $u$ and $v$, then $|p|\ge D$. Each element of $R$ has
    probability at least $D/n$ of belonging to $p$ (independently), so the
    probability that no element of $R$ belonging to $p$ is at most
    \begin{equation}\label{eq:sample_prob}
        \left(1 - \frac{D}{n}\right)^{|R|} \le \exp\left(-\frac{D}{n}\cdot
            |R|\right) \le \exp(-c\log n) = n^{-c}\ .
    \end{equation}
    Since there are at most $n^2$ such pairs, by a union bound the probability
    that there exists a pair $u,v$ with $dist_G(u,v)\ge D$, such that no
    element on a shortest path between $u$ and $v$ is sampled in $R$ is thus at
    most $n^2\cdot n^{-c} = n^{2-c}$
\end{proof}
By setting $c > 2$ we can ensure that the expected number of times we have to
re-sample the set $R$ until the condition of Lemma~\ref{lem:old_dist_probability} is
satisfied is $O(1)$. The labels can be assigned using $O(|R|\log n) =
O(\frac{n}{D}\log^2 n)$ bits as each distance can be stored using $O(\log n)$
bits.

\subsection{A scheme for medium distances}\label{sec:med_d_dist}
We now present a scheme, which preserves distances in the range $[D,2D]$
using $O\!\left(\frac{n}{D}\log^2 D\right)$ bits. More formally, we present a
labeling scheme such that given a family of unweighted graphs
$\calG$ the encoder, $e$, and the decoder, $d$, satisfies the following
constraints for any $G\in\calG$:
\begin{enumerate}
    \item $d(e_G(u), e_G(v)) \ge dist_G(u,v)$ for any $u,v\in G$.
    \item $d(e_G(u), e_G(v)) = dist_G(u,v)$ for any $u,v\in G$ with
        $dist_G(u,v)\in [D,2D]$.
\end{enumerate}
Let such a labeling scheme be called a \emph{$[D,2D]$-preserving distance
labeling scheme}.

The labeling scheme is based on a sampling procedure similar to that presented
in Section~\ref{sec:old_d_dist_scheme}, but improves the label size by
introducing the notion of \emph{sick} and \emph{healthy} nodes.
Below we described only the labeling scheme for undirected graphs. We note that
this can be turned into a labeling scheme for directed graphs by at most
doubling the label size. In our undirected labeling scheme we store distances
to several nodes in the graph, and for a directed scheme one simply needs to
store both distances \emph{to} and \emph{from} these nodes. This will be
evident from the description below.

Let $G=(V,E)\in\calG$. We sample a multiset $R$ of size $2\cdot \frac{n}{D}\log
D$. Similar to Section~\ref{sec:old_d_dist_scheme}, each element of $R$ is picked
uniformly at random from $V$.

\begin{definition}
    Let $R$ be as defined above and fix some node $u$. We say that a node $v$
    is \emph{uncovered} for $u$ if $dist_G(u,v)\ge D$ and no node in $R$ is
    contained in a shortest path between $u$ and $v$. A node $u$ with more than
    $\frac{n}{D}$ uncovered nodes is called \emph{sick} and all other
    nodes are called \emph{healthy}.
\end{definition}

Let $S$ denote the set of sick nodes and let $\uc(u)$ denote the set of
uncovered nodes for $u$. The main outline of the scheme is as follows:
\begin{enumerate}
    \item Each node $u$ stores the distance from itself to each node of $R\cup
        S$.
    \item If $u$ is healthy, $u$ stores the distance from itself to every $v\in
        \uc(u)$ for which $dist_G(u,v)\in [D,2D]$.
\end{enumerate}
We start by showing that the set of sick nodes has size $O(n/D)$ with
probability at least $1/2$. This is captured by the following lemma.
\begin{lemma}\label{lem:num_sick}
    Let $R$ be defined as above and let $S$ be the set of sick nodes. Then
    \[
        \Prp{|S| \ge 2\frac{n}{D}} \le 1/2\ .
    \]
\end{lemma}
\begin{proof}
    Fix some node $u\in V$ and let $v\in V$ be a node such that $dist_G(u,v)\ge
    D$. Using the same argument as in \eqref{eq:sample_prob} of
    Lemma~\ref{lem:old_dist_probability} we see that the probability that $v$ is
    uncovered for $u$ is at most $D^{-2}$. Therefore $\Ep{|\uc(u)|} \le
    \frac{n}{D^2}$. By Markov's inequality we have
    \[
        \Prp{u \text{ is sick}} = \Prp{|\uc(u)| \ge D\cdot \frac{n}{D^2}} \le
            \frac{1}{D}\ ,
    \]
    and thus $\Ep{|S|} \le n/D$. We again use Markov's inequality to conclude
    that
    \[
        \Prp{|S| \ge 2\frac{n}{D}} \le 1/2\ .
    \]
\end{proof}

The goal is now to store the distances to the nodes of $R\cup S$ as well as
$\uc(u)$ using few bits.
First consider the distances to the nodes of $R\cup S$. Observe that since we
only wish to recover distances in the interval $[D,2D]$ we only need to store
distances to the nodes of $R\cup S$ which are at most $2D$ away. Let $u$ be any
node in $G$. We will store the distances from $u$ to the relevant nodes of
$R\cup S$ as follows: We first fix a canonical ordering of the nodes in $R\cup
S$, which is the same for all nodes $u\in G$. For each node of $R\cup S$ in
order we now store either a 0-bit if its distance to $u$ is greater than $2D$.
Otherwise we store the distance using at most $O(\log D)$ bits.

We may now assign the label $\ell(u)$ of a node $u$ to be $id(u)$ concatenated
with the bitstring resulting from the above procedure for $R\cup S$. If $u$ is
healthy we concatenate an identifier for the set $uc(u)$ of uncovered nodes
restricted to nodes within distance $[D,2D]$ along with the distance to each of
these nodes. The decoder works by simply checking if one nodes stores the
others distance or by taking the minimum of going via any node in $R\cup S$.

\paragraph{Label size}
In order to bound the size of the label we first observe that $R\cup S$ has
size at most $O(\frac{n}{D}\log D)$ and we can thus store the distance (or a 0-bit) to
each of these nodes using $O(\frac{n}{D}\log^2 D)$ bits. We thus only need to bound the
size of storing id's and distances to the nodes of $uc(u)$ whose distance is in
$[D,2D]$. Since we only store this for healthy nodes this set has size at most
$n/D$ and can be described using at most
\[
    O\!\left(\log\!\binom{n}{n/D}\right) = O\!\left(\frac{n}{D}\log D\right)
\]
bits. Since each distance can be stored using $O(\log D)$ bits we conclude that
the total label size is bounded by $O\!\left(\frac{n}{D}\log^2 D\right)$.

\begin{theorem}\label{thm:d_dist_med}
    There exists a $[D,2D]$-preserving distance labeling scheme for the family
    $\calG_n$ of unweighted graphs on $n$ nodes with maximum
    label size
    \[
        O\!\left(\frac{n}{D}\log^2 D\right)\ .
    \]
\end{theorem}
\begin{proof}
    This is a direct corollary of the discussion above.
\end{proof}

\subsection{Bootstrapping the scheme}\label{sec:d_dist_full}
In order to show Theorem~\ref{thm:d_dist} we will concatenate several instances of the
label from Theorem~\ref{thm:d_dist_med}. First define $\ell_D(u)$ to be the
$[D,2D]$-preserving distance label for the node $u$ assigned by the scheme of
Theorem~\ref{thm:d_dist_med}. Now assign the following label to each node $u$:
\begin{equation}\label{eq:d_dist_scheme}
    \ell(u) = \ell_D(u)\ \circ\ \ell_{2D}(u)\ \circ\ \ell_{4D}(u)\ \circ\
    \ldots\ \circ\ \ell_{2^kD}(u)\ ,
\end{equation}
where $k = \floor{\lg(n/D)}$. Let $d_D$ be the distance returned by running the
decoder of Theorem~\ref{thm:d_dist_med} on the corresponding component,
$\ell_D(u)$, of the label $\ell(u)$. Then we let the decoder of the full
labeling scheme return
\begin{equation}\label{eq:d_dist_dec}
    \hat{d} = \min(d_D, d_{2D}, \ldots, d_{2^kD})\ ,
\end{equation}
with $k$ defined as above. We are now ready to prove Theorem~\ref{thm:d_dist}.

\begin{proof}[Proof of Theorem~\ref{thm:d_dist}]
    Consider any pair of nodes $u,v$ in some graph $G\in\calG_n$ and let $d
    = dist_G(u,v)$. Also, let $\hat{d}$ be the value returned by the decoder
    for $\ell(u)$ and $\ell(v)$. If $d\le D$ we have $\hat{d}\ge d$. Now assume
    that $d\in [2^i\cdot D, 2^{i+1}\cdot D]$ for some non-negative integer $i$.
    Then, by Theorem~\ref{thm:d_dist_med} and \eqref{eq:d_dist_dec} we
    have $\hat{d} = d$.

    The size of the label assigned by \eqref{eq:d_dist_scheme} is bounded by
    \begin{align*}
        \sum_{i=0}^{\floor{\lg_2(n/D)}} O\!\left(\frac{n}{2^i\cdot D}\log^2(2^i\cdot
        D)\right)
        &\le
        \sum_{i=0}^{\infty} O\!\left(\frac{n}{2^i\cdot D}\log^2(2^i\cdot
        D)\right) \\
        &\le
        O\!\left(\frac{n}{D}\log^2(D)\sum_{i=1}^{\infty} \frac{i^2+1}{2^i}\right) \\
        &=
        O\!\left(\frac{n}{D}\log^2(D)\right)
    \end{align*}
\end{proof}

\subsection{Lower bound}\label{sec:lower_bound}

\begin{proof}[Proof of Theorem~\ref{thm:lowerapprox}]

    Let $k = \floor{\frac{n}{D+1}}$ and let $L$ and $R$ be
    sets of $k$ nodes which make up the left and
    right side of a bipartite graph respectively. Furthermore, let each node of
    $R$ be the first node on a path of $D$ nodes.

    Consider now the family of all such bipartite graphs $(L,R)$ with the
    attached paths. There are exactly $2^{k^2}$ such graphs.

    Now observe, that a node $u\in L$ is adjacent to a node $v\in R$ if and
    only if $dist(u,w) = D$, where $w$ is the last node on the path starting in
    $v$. By querying all such pairs $(u,w)$ we obtain $k^2$ bits of
    information using only $2k$ labels, thus at least one label of size
    \[
        \frac{k^2}{2k} =
        \frac{\floor{\frac{n}{D+1}}}{2}
        \ge
        \frac{n}{8D}
    \]
    is needed. Since the graph has $ \le n$ nodes this implies the result.

\end{proof}
This is illustrated in Figure~\ref{fig:lower_bound}.
\begin{figure}[htbp]
    \centering
    \includegraphics[width=.5\textwidth]{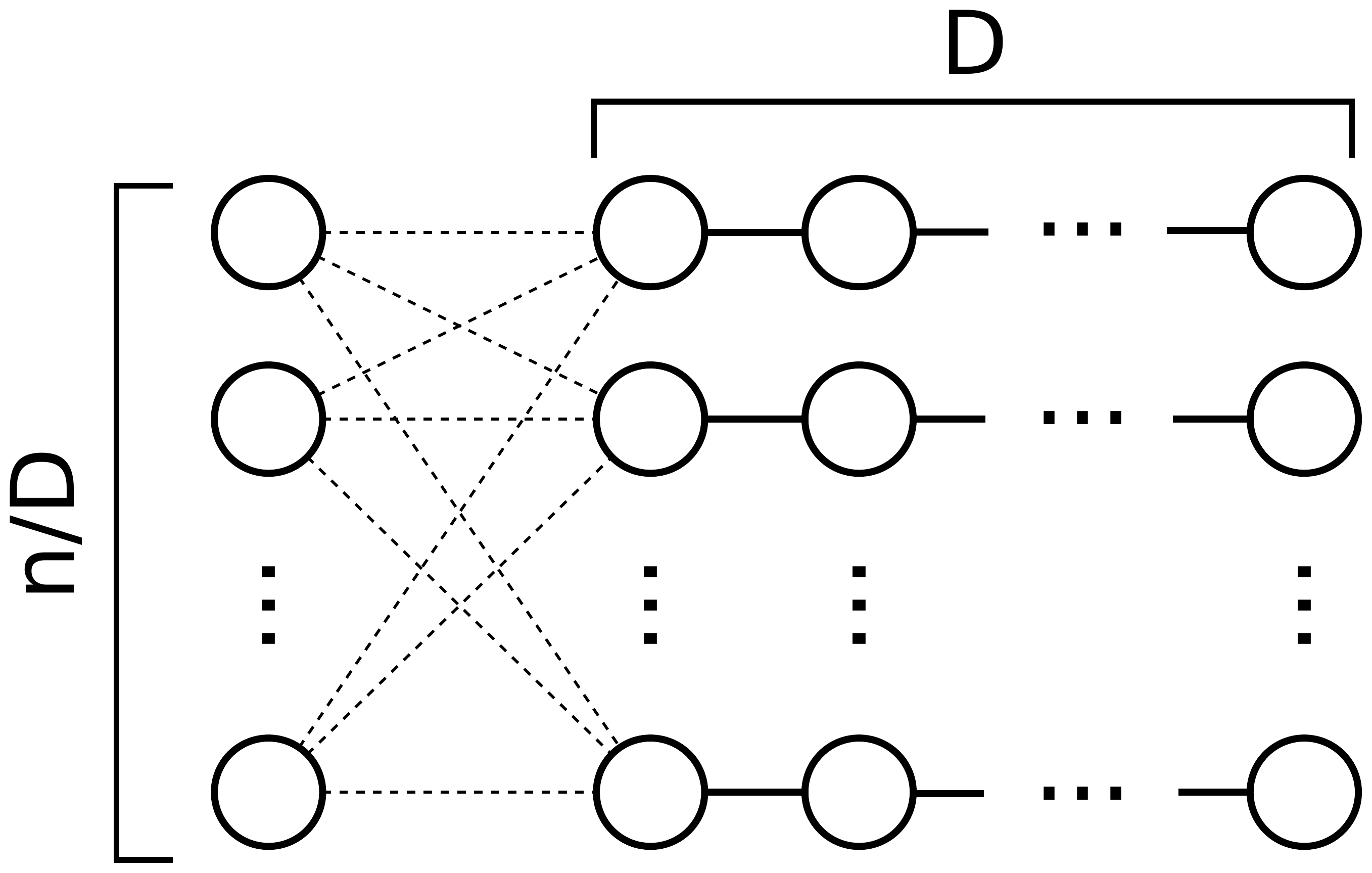}
    \caption{Illustration of the graph family used in the proof of
    Section~\ref{sec:lower_bound}.}
    \label{fig:lower_bound}
\end{figure}

\section{Sparse and bounded degree graphs}\label{sec:bound_sparse}
We are now ready to prove Theorem~\ref{thm:sparse}. In fact we will show the following
more general lemma:
\begin{lemma} \label{lem:sparse}
    Let $\mathcal{H}_{n,m}$ denote the family of unweighted
    graphs on $n$ nodes with at most $m$ edges. Then there
    exists a distance labeling scheme for $\mathcal{H}_{n,m}$ with maximum label
    size
    \[
        O\!\left ( \frac{n}{D} \cdot \log^2 D \right )
        , \ \text{where} \  \ D = \frac{\log n}{1+\log \frac{m+n}{n}}
    \]
\end{lemma}
Since $\frac{\log n}{1+\log \frac{m+n}{n}} = \omega(1)$ when $m = n^{1+o(1)}$ it
will suffice to prove Lemma~\ref{lem:sparse}. In order to do so we first show
the following lemma for bounded-degree graphs:
\begin{lemma}\label{lem:bound_deg}
    Let $\mathcal{B}_n(\Delta)$ be the family of graphs on $n$ nodes with
    maximum degree $\Delta$. There exists a distance labeling scheme for
    $\mathcal{B}_n(\Delta)$ with maximum label size
    \[
        O\!\left ( \frac{n}{D} \log^2 D \right )
        , \ \text{where} \  \ D = \frac{\log n}{1 + \log \Delta}
    \]
\end{lemma}
\begin{proof}
    Suppose we are labeling some graph $G\in\mathcal{B}_n(\Delta)$ and let
    $u\in G$. Let $D = \ceil{\frac{\log n}{1+2\log \Delta}}$ and let $\ell_D(u)$ be the
    $D$-distance preserving label assigned by using Theorem~\ref{thm:d_dist} with
    parameter $D$. Using this label we can deduce the distance to all nodes of
    distance at least $D$ to $u$.

    Since $G\in\mathcal{B}_n(\Delta)$ there are at most $\Delta^D = O\!\left(\sqrt{n}\right)$
    nodes closer than distance $D$ to $u$. Thus, we may describe the IDs
    and distances of these nodes using at most $O(\sqrt{n}\log n)$ bits. This
    gives the desired total label size of
    \[
        |\ell(u)| = O\!\left(\sqrt{n}\log n + \frac{n}{D}\log^2 D\right)
        =O\!\left(\frac{n}{D}\log^2 D\right)
    \]
\end{proof}

Using this result we may now prove Lemma~\ref{lem:sparse} by reducing to the
bounded degree case in Lemma~\ref{lem:bound_deg}. This has been done before e.g.
in distance oracles \cite{elkinPettie15,agarwalGH11}.
\begin{proof}[Proof of Lemma~\ref{lem:sparse}]
    Let $G\in\mathcal{H}_{n,m}$ be some graph and let
    $k = \max\left\{\ceil{\frac{m}{n}},3\right\}$. Let $u\in G$ be some node with more than
    $k$ incident edges. If no such node exists, we may apply
    Lemma~\ref{lem:bound_deg} directly and we are done. Otherwise we split $u$ into
    $\ceil{\deg(u)/(k-2)}$ nodes and connect these nodes with a path of
    $0$-weight edges. Denote these nodes $u^1,\ldots,u^{\ceil{\deg(u)/(k-2)}}$.
    For each edge $(u,v)$ in $G$ we assign the end-point at $u$ to a node
    $u^i$ with $\deg(u^i) < k$. This process is illustrated in
    Figure~\ref{fig:sparse_trans}.

    Let the graph resulting from performing this process for every node $u\in
    G$ be denoted by $G'$. We then have $\Delta(G') \le k$.
    Furthermore it holds that for every pair of nodes $u,v\in G$ we have
    $dist_G(u,v) = dist_{G'}(u^1,v^1)$. Consider now using the labeling scheme
    of Lemma~\ref{lem:bound_deg} on $G'$ and setting $\ell(u) = \ell(u^1)$ for each
    node $u\in G$. We note that splitting nodes in the graph results in a
    weighted graph with weights $0$ and $1$. However, one can observe that the
    labeling scheme of Theorem~\ref{thm:d_dist} actually preserves distances
    for nodes who have at least $D$ \emph{edges} on a shortest path between
    them. It thus follows that this is actually a distance labeling scheme for
    $G$. The number of nodes in $G'$ is bounded by
    \[
        \sum_{u \in G} \ceil{\frac{\deg(u)}{k-2}} \le
        \sum_{u \in G} \left ( \frac{\deg(u)}{k-2} +1 \right ) =
        \frac{2m}{k-2} + n = O\!\left(n\right) \ ,
    \]
    which means that Lemma~\ref{lem:bound_deg} gives the desired label size.
\end{proof}

\begin{figure}[htbp]
    \centering
    \includegraphics[width=.5\textwidth]{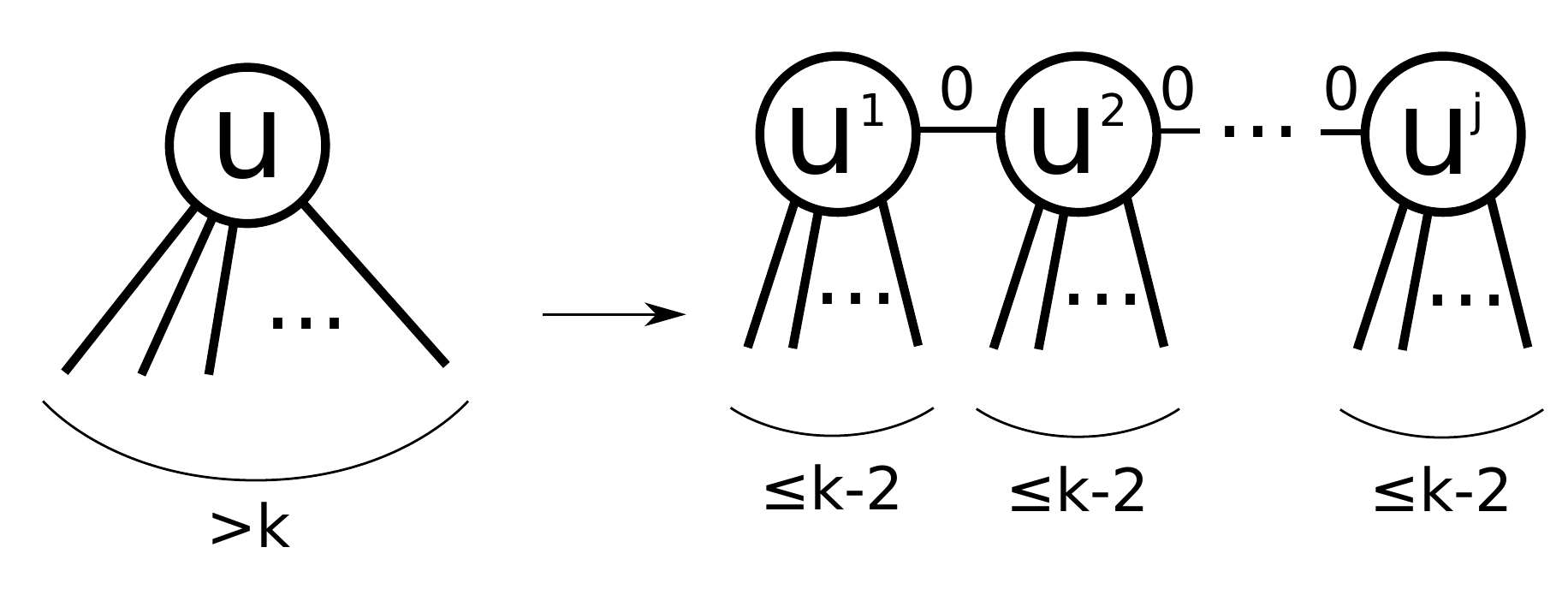}
    \caption{Illustration of the transformation from sparse graph to bounded
    degree graph.}
    \label{fig:sparse_trans}
\end{figure}

\section{Additive error}\label{sec:additive}
We will now show how we can use our $D$-preserving labeling scheme of
Theorem~\ref{thm:d_dist} to generalize the $2$-additive distance labeling
scheme of Gawrychowski et al.~\cite{GawrychowskiKUEvenSimpler}. We will
assume that $r \le n^{1/10}$ for simplicity.

Let $t = r\log^{10} n$
and let $D = \frac{r\log n}{4\log t}$. We describe the scheme in three parts:
\begin{enumerate}
    \item Let $G^r$ be a copy of $G$, where an edge is added between any pair
        of nodes whose distance is at most $r/2$ in $G$. Let $V^r_{\ge t}$ be the
        set of nodes in $G^r$ with degree at least $t$ and let $S$ be
        a minimum dominating set of $V^r_{\ge t}$ in $G_r$. Then $|S| =
        O(\frac{n\log t}{t})$.

        \medskip
        For all nodes $u\in G$ we store $dist(u,v)$ and $id(v)$ for all $v\in
        S$.
    \item Consider now the subgraph of $G$ induced by $V \setminus V^r_{\ge
        t}$. For a node $u\notin V^r_{\ge t}$, let $B_u(D)$ be the ball of
        radius $D$ around $u$ in this induced subgraph Then $|B_u(D)|\le
        t^{2D/r} = O(\sqrt{n})$. This follows from the definition of $V^r_{\ge
        t}$: There are at most $t$ nodes within distance $r/2$ from $u$ and thus
        at most $t^2$ nodes within distance $r$ from $u$, etc.

        \medskip
        For all $u\notin V^r_{\ge t}$ we store $dist(u,v)$ and $id(v)$ for all
        $v\in B_u(D)$.
    \item Finally we store a $D$-preserving distance label for all
        $u\in G$.
\end{enumerate}

The total label size is then
\[
    O\left(
        n\cdot\frac{\log n\log t}{t}
        + \sqrt{n}\log n
        + \frac{n\log t}{r\log n}\cdot (\log(r\log n))^2
    \right)
    =
    O\left(
        \frac{n}{r\log n}\cdot\polylog(\log n \cdot r)
    \right)\ ,
\]
as stated in Theorem~\ref{thm:r_approx}.

\paragraph{Decoding}
To see that the distance between two nodes $u$ and $v$ can be calculated within
an additive error $r$ we split into several cases:
\begin{itemize}
    \item If $dist(u,v) \ge D$ we can report the exact distance between $u$ and
        $v$ using the $D$-preserving distance scheme.
    \item If $dist(u,v) \le D$ and $deg_{G_r}(v) \ge t$ we can find a node
        $z\in S$ such that $dist(z,v) \le r/2$ and thus
        \[
            dist(u,z) + dist(z,v) \le dist(u,v) + dist(v,z) + dist(z,v) \le
            dist(u,v) + r\ ,
        \]
        and symmetrically if $deg_{G_r}(u)\ge t$.
    \item Finally, if $dist(u,v)\le D$ and $deg_{G_r}(u) < t$ and $deg_{G_r}(v)
        < t$, then we $v\in B_u(D)$ and we can thus report the exact distance
        between $u$ and $v$.
\end{itemize}

\paragraph{Acknowledgements}
We would like to thank Noy Rotbart for helpful discussions and observations.


\bibliographystyle{plain}
\bibliography{distlabel}

\end{document}